\documentclass[journal]{IEEEtran}

\usepackage{cite}
\usepackage{float}

\ifCLASSINFOpdf
\else
\fi
\usepackage{amsthm}
\usepackage{amsmath}
\usepackage{graphicx}
\usepackage{cite}
\newtheorem{definition}{Definition}
\newtheorem{theorem}{Theorem}
\newtheorem{lemma}{Lemma}
\newtheorem{remark}{Remark}
\begin{document}
%
\title{Optimal Demand Response and Supply Schedule Under Market-Driven Price: A Stackelberg Dynamic Game Approach}
%
%
%

\author{Yunhan Huang
\thanks{Yunhan Huang is with the Department
of Electrical and Computer Engineering, New York University, USA.}
\thanks{The manuscipt is written in October, 2017.}}

%
%

\markboth{arXiv}%
{Shell \MakeLowercase{\textit{et al.}}: Bare Demo of IEEEtran.cls for IEEE Journals}
%



\maketitle

\begin{abstract}
In this work, we use a Stackelberg infinite discrete-time dynamic game model to study the optimal supply schedule and the optimal demand response under a market-driven dynamic price. A two-layer optimization framework is established. At the lower layer, for each user, different appliances are scheduled for energy consumption. For enegy provider, different generators are utilized for energy generation. At the upper level, with the supplier acting as a leader and the users acting as followers, a Stackelberg dynamic game is used to capture the interaction among the energy provider and the users where the energy provider and the users care only about their own cost. We analyze the one-leader-N-followers Stackelberg dynamic game and characterize the Stackelberg equilibrium. We provide a closed-form Nash solution of the optimal dynamic demand response problem when the supply is announced. A set of linear constraints is developed to characterize the Stackeberg equilibrium. Simulation results show that the price is driven to a reasonable value. Also, the total demand and the supply is balanced. 
\end{abstract}


%
\IEEEpeerreviewmaketitle

\section{Introduction}
%
%
%
%
Equipped with modern information, communication and electronics technology, the goal of providing reliable, efficient, secure, and quality energy generation/distribution/consumption in power grid becomes achievable. People uses smart grid to describe this kind of next generation electrical power grid. Developing smart grid has become an urgent global priority as its economic, environmental, and societal benefit will be enjoyed by generations to come. 

Smart grid enables the end-customers and the suppliers to participate into the market to provide a cost-effective and reliable energy demand and supply. The reliable and efficient communication infrastructure plays a key role in connecting the supply side and the demand side and managing, controlling, optimizing different devices and systems in the smart grid. In the demand side, people depends on demand side management to minimize the power cost, reduce the fluctuation of power grid \cite{Mohsenian10,Quanyan12,Kamyab16}. Direct load control and smart pricing are two popular approaches in demand side management. For smart pricing, researchers focus on how to design the price for manager/authority/utility company to reduce consumption at peak hours, attract more customers, maximize their profit \cite{Paola17,Ghasemkhani18,Li11}. Usually, the price is set to be a function of the total demand of all users. There is little work in the literature, which considers the price that is driven by the market. Market-driven dynamic price is determined by not just the current total demand but also the supply, the past demand and the past price. Besides, many works study just the demand side management and simplify the effects of supply side \cite{Quanyan12,Juntao17}. However, the interaction between the demand side and the supply side is non-negligible.

In this paper, we jointly consider the optimal demand response and the optimal supply schedule under the market-driven dynamic price. The market price is modeled by the sticky price dynamics where the current price is determined by demand, supply, past price. Moreover, to characterize the price volatility, a stochastic random variable is added to the dynamics. The supplier in the market aims to minimize his cost (maximize his profit) by schedule his supply based on the knowledge on the price dynamics. Users in the market minimize their cost by selecting an optimal demand response based on the knowledge on the price dynamics and the scheduled supply. The interactions between the price, the demand and the supply are intertwined and hart to model. Here, we utilize a stochastic discrete-time infinite dynamic game with finite horizon to model this scenario and the Stackelberg equilibrium solution \cite{Maharjan13, JuntaoSt17} is used to characterize the optimal demand response and the optimal supply schedule. We obtain a closed-form Nash equilibrium solution which refers to the optimal demand response corresponding to a given supply schedule. Further, a set of linear constraints are developed to characterize the Stackelberg equilibrium. We call the problem modeled by the stochastic dynamic game the upper level problem. In the lower level, for each user, an optimal power allocation problem is solved individually for scheduling different appliances to maximize the utility. On the supply side, the supplier allocates the work of generation to different generators. The lower level problems are modeled by a set of decoupled convex/concave static optimization problems. In the simulation results, we present the evolution of the market-driven dynamic price and the corresponding optimal demand response and optimal supply schedule. The results show that the demand and the supply can be tuned to be balanced by the dynamic price.

The remainder of the paper is structured as follows. The two layer optimization framework is formulated in Section \ref{ProFor}. In Section \ref{Sec:Ana}, we analyze the proposed two-level optimization and dynamic game model. A closed-form Nash solution of the users corresponding to supply schedule is obtained in a closed-form. A set of linear constraints are developed to characterize the Stackelberg equilibrium. Simulation results are studied in Section \ref{Sec:Sim}. Conclusions are included in Section \ref{Sec:Con}.

\section{Problem Formulation}\label{ProFor}
In this section, we propose a Stackelberg dynamic game model for optimal demand response and supply schedule in the smart grid. The framework is comprised of two levels of problems as illustrated in Fig. \ref{fig:Framework}. At the upper level, the energy supplier announces the energy supply schedule. The energy supply is scheduled to minimize the energy supplier's cost. For users, the demand response management device of each user purchases power from the market based on its demand and the price. At the lower level, the energy supplier allocates the work of generation to multiple generators like wind driven generator, solar generator etc. The scheduler of each user allocates its demand power to multiple appliances such as the refrigerator, the plug-in hybrid electric vehicle (PHEV), microwaves etc. At the upper level, the decisions are made by solving hierarchical optimization problem where the supplier announces its supply while the schedulers of the users decide their demands. Both the supplier and the users aim to minimize their own cost based on the market-driven dynamic price. So, at the upper level, the supplier and the users constitute the leader and the followers of a stackerlberg dynamic game. While at the lower level, the optimal power allocation within each user is decoupled from other users. And the supplier arranges its generation independently. Hence, the decisions at the lower level can be modeled by static convex optimization models.

\begin{figure}
    \centering
    \includegraphics[width=\linewidth]{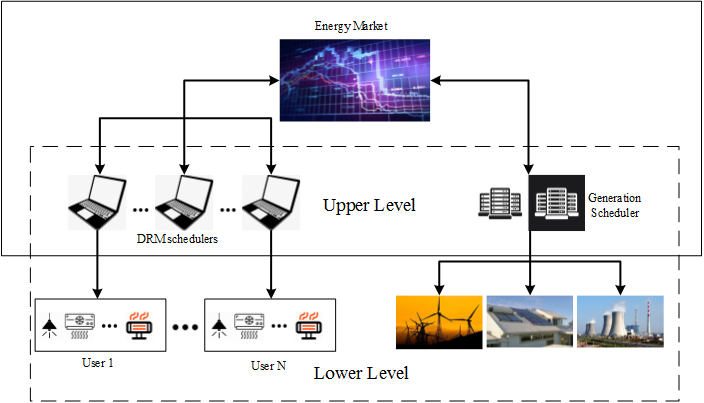}
    \caption{Framework: Upper level demand response and supply schedule problem and lower level consumption and generation allocation problem.}
    \label{fig:Framework}
\end{figure}

\subsection{Upper Level: A Stackelberg Dynamic Game Model}
Let $l$ denote the power supplier and $\mathcal{N}$ be the set of users or households in the power grid. Let $s \geq 0$ be the power supply in the grid which is decided by the energy supplier. And $d_t^i\geq0, i\in \mathcal{N}$ denotes the demand of the $i$-th player at time $t\in \mathbf{Z}_{+}$ in the grid. The power price $p_t$ is a market-driven dynamic price which evolves according to the following dynamics \cite{Brigo03, DeJong06, Hou05}:
\begin{equation}\label{PriceDynamics}
p_{t+1}=(1+\omega_t)p_t+f(d_t,s_t)+\theta_t
\end{equation}
with initial condition $p_1>0$, where $\omega_t<0$ is the parameter that suggests the stickiness of the price in response to the demand and supply; $d_t=\sum_{i\in\mathcal{N}} d_t^i$; $\theta_t,t\in \mathbf{Z}_+$ is a set of statistically independent Gaussian random numbers with zero mean, and $cov(\theta_t)>0, \forall t \in \mathbf{Z}_+$, which models the volatility in the market \cite{Brigo03}; $f(\cdot,\cdot):\mathbf{R}_+ \times \mathbf{R}_+ \rightarrow \mathbf{R}$ is a force function that drives the price which is based on the total demand $d_t$ and the supply $s_t$. The price is driven up if the demand is high but is driven down when the supply is high. The price at time $t+1$, $p_{t+1}$, is based on by the price, the total demand and the supply at time $t$, i.e., $p_t$, $d_t$ and $s_t$. This is due to the price delay: the delay with which the energy response to information \cite{Hou05}. The price could be negative .We let $f(\cdot,\cdot)$ take the following form:
\begin{equation}\label{spf}
    f(d_t,s_t)=\gamma(d_t-s_t).
\end{equation}

The objective of user $i$ is to minimize his own cost by deciding the demand $d_t^i$ at each time slot $t$ based on the evolution of the price and his own utility. We utilize the following cost functional to characterize the cost of user $i$ during time $1$ to $T$:
\begin{equation}\label{CostFollows}
    C^i(d^i,\mathbf{d}^{-i}, \mathbf{s}, p_t)=\mathbf{E}\sum\limits_{t=1}^{T} -u^i(t,d_t^i)+p_t d_t^i +\frac{1}{2} \alpha_t^i {d_t^i}^2
\end{equation}
where $\mathbf{d}^{-i}=\{ d^j,j\neq i,j\in\mathcal{N} \}$; $u^i(\cdot,\cdot):\mathbf{T}\times \mathbf{R}_+\rightarrow \mathbf{R}_+$ is the stage utility function that describes the utility obtained by consuming energy $d_t^i$, where $\mathbf{T}=\{1,2,...,T\}$ denotes the set of time slots. The second term characterizes the payment of user $i$ consuming $d_t^i$ unit of power. The quadratic term $\frac{1}{2} \alpha_t^i {d_t^i}^2$ is the user $i$'s adjustment cost on power usage where $\alpha_t^i$ is positive for all $i\in\mathcal{N}$ and $t\in\mathbf{T}$. The supply schedule $\mathbf{s}=\{s_t,t\in\mathbf{T}\}$.

The aim of the energy supplier is to minimize his own cost which is written as follows:
\begin{equation}\label{CostLeader}
    C^l(\mathbf{s},\mathbf{d},p_t)=\mathbf{E}\sum\limits_{t=1}^T \beta(t,s_t)-p_t s_t + \frac{1}{2}\kappa_t(s_t-d_t)^2
\end{equation}
where $\beta(\cdot,\cdot):\mathbf{T} \times \mathbf{R}_+\rightarrow \mathbf{R}_+$ is the generation cost that describes the cost of generating $s_t$ unit of power at each time slot $t$. The second term is the profit of selling $s_t$ unit of power and the last one represents a penalty to the leader who as a leader has the responsibility to balance the supply and demand where $\kappa_t\in\mathbf{R}_++$. 

The energy supplier who acts as a leader announces his generation strategy to the market. The users who react rationally to the leader's decision can be called followers. Both the energy supplier and the users aim to minimize their own cost subject to the evolution of the price which constitutes a one-leader-N-followers Stackelberg dynamic game. In this game, the leader aims to minimize the cost functional (\ref{CostLeader}) and the followers minimize the cost functional (\ref{CostFollows}) subject to the price dynamics (\ref{PriceDynamics}).

\subsection{Lower Level: Optimal Local Allocation}
In this part, the discuss how the users and the energy supplier allocates their consumption and generation at the lower level.
\subsubsection{Optimal Generation Allocation}
The energy supplier is equipped with multiple resources for power generation. For example, in a microgrid, the supplier can utilize wind generator, solar generator and thermal generator. Let $\mathcal{G}$ be the set of $G$ generators for the supplier. At each time $t$, the power generated by generator $g$ is $e^g_t$. We define an energy generation scheduling vector $\mathbf{e}_t=[e_t^1, e_t^2,...,e_t^G] \in \mathbf{R}_+^G$. The supplier seeks to allocate power generation task to its generators by minimizing the total cost subject to the power generation constraint $s_t$. This is described by the the generation optimization problem (GO) as follows.
\begin{equation}
\begin{aligned}
    \textrm{(GO)}\ \min \limits_{\mathbf{e}_t \in \mathbf{R}_+^G} &\sum\limits_{g\in\mathcal{G}} \frac{1}{2} \delta_t^g {e_t^g}^2 \\
    \textrm{s.t.} & \sum\limits_{g\in\mathcal{G}}e_t^g =s_t,
\end{aligned}
\end{equation}
where $\delta_t^g$ denotes the cost parameter of generator $g$ at time $t$. $\frac{1}{2} \delta_t^g {e_t^g}^2$ is the cost of generator $g$ generating $e_t^g$ unit of power. At different time, generators have different cost parameters. For example, the cost parameter of solar generator is low during daytime while the cost parameter of oil-fed generator won't change too much over time.

The optimization problem (GO) is a convex optimization problem whose solution can be easily obtained \cite{Luenberger97}:
\begin{equation}
    \lambda_t = \frac{-s_t}{\sum_{g\in\mathcal{G}}1/\delta_t^g},\textrm{    }e_t^g=\frac{s_t}{\delta^g_t \sum_{g\in\mathcal{G}}1/\delta_t^g} 
\end{equation}
where $\lambda_t$ is the Lagrange multiplier. Based on the optimal solution, we derive the following
\begin{equation}
\beta(t,s_t)=\frac{1}{2}\bar{\delta}_t s_t^2,
\end{equation}
where $\bar{\delta}_t= \sum_{g\in\mathcal{G}}\frac{1}{\delta_t^g}(\frac{1}{ \sum_{g\in\mathcal{G}}1/\delta_t^g})^2$.

\subsubsection{Optimal Consumption Allocation}
Each user $i$ optimizes its power usage by effectively allocating power to appliances. Let $\mathcal{A}^i=\{a^i_1, a^i_2,...,a^i_{A^i}\}$ be the set of $A^i$ appliances for user $i$ such as air-conditioner, heater, dryer, PHEV etc. Define an energy consumption scheduling vector $\mathbf{v}^i_t=[v_t^{i,1},v_t^{i,2},...,v_t^{i,A^i}] \in \mathbf{R}_+^{A^i}$, where $v_t^{i,a}$ denotes the power consumption that is scheduled for appliance $a \in \mathcal{A}^i$ of user $i$. Each user $i$ seeks to allocate power to its appliances by maximizing their their total utilities. This can be written as the consumption optimization problem (CO) as follows.
\begin{equation}
\begin{aligned}
    \textrm{(CO)}\ \max\limits_{\mathbf{v}_t^i \in \mathbf{R}_+^{A^i}} &L^i(\mathbf{v}_t^i)=\exp \Big\{\sigma^i \sum_{a^i \in \mathcal{A}^i} \psi_t^{i,a} \ln (v_t^{i,a}) \Big\}\\
    \textrm{s.t.} & \sum\limits_{a\in\mathcal{A}^i} v_t^{i,a} =d_t^i,
\end{aligned}
\end{equation}
where $\sigma^i\geq 0$ is a sensitivity parameter with respect to the demand;  $\psi_t^{i,a}\geq 0, a\in\mathcal{A}^i$  is a parameter that indicates the importance and time-sensitivity of each appliance $a \in \mathcal{A}^i$. Without loss of generality, we assume that $\sum_{a\in\mathcal{A}^i} \psi_t^{i,a}=q$. The objective function $L^i: \mathbf{R}_+^{A^i}\rightarrow \mathbf{R}_+$ can be equivalently written into the form of products, i.e., $L^i(\mathbf{v}_t^{i,a})=\prod_{a\in\mathcal{A}^i} ({v_t^{i,a}})^{\psi_t^{i,a}\sigma^i}$. Hence, the objective function is analogous to Nash bargaining problems \cite{Yaiche00} and the outcome of its solution embodies proportional fairness \cite{Kelly98}. Due to the monotonicity of the exponential function, (CO) is related to the following equivalent consumption optimization problem (ECO):
\begin{equation}
\begin{aligned}
    \textrm{(ECO)} \max\limits_{\mathbf{v}_t^i \in \mathbf{R}_+^{A^i}} &\tilde{L}^i(\mathbf{v}_t^i) := \sum\limits_{a \in \mathcal{A}^i} \psi_t^{i,a}\ln(v_t^{i,a})\\
    \textrm{s.t.}& \sum\limits_{a\in\mathcal{A}^i} v_t^{i,a} =d_t^i,
\end{aligned}
\end{equation}
where $\tilde{L}^i: \mathbf{R}_+^{A^i} \rightarrow \mathbf{R}$ and $L^i=\exp(\tilde{L}^i)$. The optimization problem (ECO) is concave. So we obtain the optimal solution as follows.
\begin{equation}
    v_t^{i,a}=\psi_t^{i,a} d_t^i,\textrm{     }\lambda_t^i=1/d_t^i.
\end{equation}
The utility value of user $i$ is the optimal value of the optimization problem (CO). Hence, the utility value can be obtained as
\begin{equation}
    u^i(t,d_t^i)= \bar{\psi}_t^{i} {d_t^i}^{\sigma^i q},
\end{equation}
where $\bar{\psi}_t^{i}=\exp(\sum_{a\in \mathcal{A}^i}\psi_t^{i,a} \ln \psi_t^{i,a})\in\mathbf{R}_+$.

\section{Analysis of the Stackelberg Dynamic Game Model}\label{Sec:Ana}
The lower level decoupled optimization problem together with the upper level Stackelberg dynamic game can be studied jointly as a stochastic dynamic game with Stackelberg solution. In this section, we analyze the stochastic dynamic game and characterize the Stackelberg equilibrium solutions.

\subsection{Generalized Linear Quadratic Dynamic Games}
With the results of problems (CO) and let $\sigma^i=1/q$ for all $i\in\mathcal{N}$, the cost of user $i$ can be specified:

\begin{equation}\label{UserCost}
        C^i(d^i,\mathbf{d}^{-i}, \mathbf{s},p_t)=\mathbf{E}\sum\limits_{t=1}^{T} -\bar{\psi}_t^i {d_t^i} +p_t d_t^i +\frac{1}{2} \alpha_t^i {d_t^i}^2.
\end{equation}

For the supplier, with the optimal generation allocation, his cost can be expressed as
\begin{equation}\label{SupplierCost}
    C^l(\mathbf{s},\mathbf{d},p_t)=\mathbf{E}\sum\limits_{t=1}^T \frac{1}{2}\bar{\delta}_t s_t^2-s_t p_t + \frac{1}{2} \kappa_t (s_t-d_t)^2 .
\end{equation}

The cost functionals (\ref{UserCost}) and (\ref{SupplierCost}) together with the dynamics of the price (\ref{PriceDynamics}) and (\ref{spf}) from a generalized stochastic linear quadratic dynamics game. To characterize the Stackelberg equilibrium of the game with open-loop information structure for both leader and followers \cite{TBBook}, the concept of the open-loop Stackelberg equilibrium solution is defined below.

\begin{definition}
For a dynamic game in the form of (\ref{UserCost}), (\ref{SupplierCost}) and (\ref{spf}), the strategy set $\{\mathbf{s}^*,\{\mathbf{d}^{i*},i\in\mathcal{N} \}\}$ with the corresponding price $p_t^*$ constitutes an open-loop Stackelberg equilibrium with the supplier acting as the leader and the users acting as the followers if the following conditions are satisfied.
\begin{enumerate}
    \item $\{\mathbf{d}^{i*}\in\mathbf{R}_+^{T},i\in\mathcal{N} \}$ is a Nash equilibrium given $\mathbf{s}^*$, i.e., $\forall i \in\mathcal{N}$, we have
    \[C^i(\mathbf{d}^{i*},\mathbf{d}^{-i*},\mathbf{s}^*,p_t^*)\leq C^i(\mathbf{d}^i,\mathbf{d}^{-i*},\mathbf{s}^*,p_t),\ \ \  \forall \mathbf{d}^i \in \mathbf{R}_+^T.
    \]
    \item $\{ \mathbf{s}^*\in\mathbf{R}_+^T \}$ satisfies
    \[
        C^l(\mathbf{s}^*,\Gamma(\mathbf{s}^*),p_t^*)\leq C^l(\mathbf{s},\Gamma(\mathbf{s}),p_t),\ \ \ \forall \mathbf{s}\in\mathbf{R}_+^T,
    \]
    where $\Gamma(\mathbf{s})$ denotes the optimal response of the users, i.e., the Nash equilibrium response given $\mathbf{s}$.
\end{enumerate}
\end{definition}
To deal with the randomness in the linear-quadratic dynamic game, we have the following theorem.
\begin{theorem}
The set of open-loop stackelberg equilibrium solutions for the stochastic linear quadratic dynamics game (\ref{UserCost}), (\ref{SupplierCost}) and (\ref{spf}) corresponds with the set of open-loop Stackelberg solutions for its deterministic counterpart, i.e., in lieu of (\ref{PriceDynamics}), the price evolves based on the following rule
\begin{equation}\label{PriceDynamics2}
p_{t+1}=F_t(p_t,d_t,s_t),\ \ \ p_{t=1}=p_1
\end{equation}
where $F_t(p_t,d_t,s_t)=(1+\omega_t)p_t+\gamma(d_t-s_t)$ and in lieu of (\ref{UserCost}) and (\ref{SupplierCost}), the cost functionals are deterministic in the following forms:
    \begin{align}
    \label{UserCost2}
    &C^i(d^i,\mathbf{d}^{-i}, \mathbf{s},p_t)=\sum\limits_{t=1}^{T} -\bar{\psi}_t^i {d_t^i} +p_t d_t^i +\frac{1}{2} \alpha_t^i {d_t^i}^2,\\ 
    \label{SupplierCost2}
    &C^l(\mathbf{s},\mathbf{d},p_t)=\sum\limits_{t=1}^T \frac{1}{2}\bar{\delta}_t s_t^2-s_t p_t + \frac{1}{2}\kappa_t (s_t-d_t)^2 .
    \end{align}
\end{theorem}
\begin{proof}
For linear-quadratic stochastic dynamic games and under the open-loop information structure, if the set $\{\theta_1,...,\theta_T\}$ is statistically independent of $p_1$, which was our assumption previous section, then the stochastic contribution completely separates out \cite{TBBook}. Also, we have $\mathbf{E}[\theta_t]=0,\forall t\in\mathbf{T}$. Hence, the open-loop Stackelberg matches completely with its deterministic counterpart.
\end{proof}
To obtain the Stackelberg equilibrium, we need to characterize the optimal response of the followers which is a Nash equilibrium when $\mathbf{s}$ is given. Given $s_t, t\in\mathbf{T}$, the optimal response can be obtained by solving a affine-quadratic dynamic game. One way of obtaining the Nash equilibrium solution is to view the affine-quadratic dynamic game as static infinite game by backward recursive substitution of (\ref{PriceDynamics2}) into (\ref{UserCost2}). The static infinite game is with cost function $\hat{C}^i(\mathbf{d}^i,\mathbf{d}^{-i})$ for $i\in \mathcal{N}$. The existence of pure Nash equilibrium for this equivalent static game can be guaranteed \cite{TBBook} and the uniqueness is guaranteed if some conditions on $\gamma$, $\alpha_t^i$ and $\omega_t$ for $t\in\mathbf{T}$ are satisfied. These conditions can be easily obtained by discussing the static quadratic game formed from the affine-quadratic dynamic game. Here, we assume that the Nash equilibrium of the affine-quadratic game is unique. Another way is to utilize the Pontryagin's maximum principle method which gives the following lemma.

\begin{lemma}\label{NashNec}
To any announced supply strategy $s_t\in\mathbf{R}_+$ for $t\in\mathbf{T}$, there exists an optimal response of the followers denoted by $\{\hat{d}^i_t,i\in\mathcal{N}, t\in\mathbf{T}\}$ satisfying the following relations: for all $i\in\mathcal{N}$ and $t\in\mathbf{T}$
\begin{align}
\label{NashNecDy}
\hat{p}_{t+1}=(1+\omega_t)\hat{p}_t+\gamma(\sum\limits_{i\in\mathcal{N}} \hat{d}_t^i - s_t),\ \ \ \hat{p}_1=p_1,\\
\label{OptimalResponse}
\hat{d}_t^{i}=\max\Big\{ \frac{\bar{\psi}_t^i-\hat{p}_t-\lambda_{t+1}^i\gamma}{\alpha_i^t},0 \Big\},\\
\label{CostateDynamics}
\lambda_t^i=(1+\omega_t)\lambda_{t+1}^i+\hat{d}_t^i,\ \ \ \lambda_{T+1}^i=0.
\end{align}
Here, $\{\lambda_1^i,...,\lambda_{T+1}^i\}$ for $i\in\mathcal{N}$ is a sequence of costate scalar associated with the affine-quadratic dynamic game among users.
\end{lemma}
\begin{proof}
Utilize the definition of the Nash equilibrium in dynamic game \cite{TBBook}. The result follows directly from the minimum principle for discrete-time control systems \cite{TBBook, Bertsekas05}.
\end{proof}
\begin{remark}
From (\ref{OptimalResponse}) and (\ref{CostateDynamics}), we can infer that $\lambda_t$ is nonincreasing over $t$ which indicates the users prefer to postpone their consumption. If the $\bar{\psi}_t^i$ is sufficiently large, $\hat{d}_t^i$ is positive. If $\hat{d}_t^i \in \mathbf{R}_{++}$ for all $i\in\mathcal{N}$ and $t\in\mathbf{R}$, the conditions in Lemma \ref{NashNec} can be further developed.
\end{remark}

\begin{theorem}\label{NashSol}
Given $s_t,t\in\mathbf{T}$, if $\hat{d}_t^i$ in Lemma \ref{NashNec} is an inner-point solution for all $i\in\mathcal{N}$ and $t\in\mathbf{T}$, the Nash equilibrium $\{\hat{d}_t^i,i\in\mathcal{N},t\in\mathbf{T}\}$ of the affine-quadratic dynamic game and the corresponding dynamics are given by
\begin{equation}\label{NashSolState}
\begin{aligned}
\hat{p}_{t+1}&=(1+\omega_t)\hat{p}_t+\gamma(\sum\limits_{i\in\mathcal{N}} \hat{d}_t^i - s_t),\ \ \ \hat{p}_1=p_1,\\
\end{aligned}
\end{equation}

\begin{equation}\label{NashSolCon}
\begin{aligned}
\hat{d}_t^i&=\frac{1}{\alpha_t^i}(\frac{q_t^i-1}{\alpha_t^i(1+\omega_t-\gamma/\alpha_t^i)}-1)\hat{p}_t\\\
&+\frac{\bar{\psi}_t^i}{\alpha_t^i}(\frac{1}{\alpha_t^i(1+\omega_t-\gamma/\alpha_t^i)}+1)\\
&+\frac{b_t^i}{{\alpha_t^i}^2(1+\omega_t-\gamma/\alpha_t^i)},
\end{aligned}
\end{equation}
for all $i\in\mathcal{N}$ and $t\in\mathbf{T}$, where
\begin{equation}\label{RiccatiLike}
\begin{aligned}
q_t^i&=\frac{(1+\omega_t-\gamma/\alpha_t^i)(1+\omega_t-\gamma\sum_{j\in\mathcal{N}}1/\alpha_t^j)}{1+\gamma^2\sum_{j\in\mathcal{N}} q_{t+1}^j/\alpha_t^j }q_{t+1}^i\\
&-1/\alpha_t^i,\ \ \ q_{T+1}=0;\\
\end{aligned}
\end{equation}

\begin{equation}\label{RiccatiLike2}
\begin{aligned}
b_t^i&=(1+\omega_t-1/\alpha_t^i)b_{t+1}^i+\bar{\psi_t^i}/\alpha_t^i\\
&+\frac{\gamma(\sum_{j\in\mathcal{N}}\bar{\psi}_t^j/\alpha_t^j-s_t-\gamma\sum_{j\in\mathcal{N}}b_{t+1}^j/\alpha_t^j)}{1+\gamma^2\sum_{j\in\mathcal{N}}q_{t+1}^j /\alpha_t^j},\\
&b_{T+1}=0,
\end{aligned}
\end{equation}
for all $i\in\mathcal{N}$ and $t\in\mathbf{T}$.
\end{theorem}
\begin{proof}
Given the set of relations in Lemma \ref{NashNec}, let $\lambda^i_t$ be in the following form $\lambda_t^i=q_t^i p_t+ b_t^i$. Substitute the $\lambda_t^i$ and $\lambda_{t+1}^i$ in (\ref{OptimalResponse}) (\ref{CostateDynamics}) by $q_t^i p_t+ b_t^i$ and $q_{t+1}^i p_t+ b_{t+1}^i$ respectively, then pulg equation (\ref{OptimalResponse}) into (\ref{NashNecDy}). Finally, combine equation (\ref{NashNecDy}) and (\ref{CostateDynamics}) to obtain the above presented closed-form solution.
\end{proof}
\begin{remark}
Theorem \ref{NashSol} provides an efficient way to compute the Nash equilibrium of the affine-quadratic dynamic game given supply schedule $\mathbf{s}$. This result is a stepping stone for characterizing the Stackelberg equilibrium of the original problem. Compared with \cite{Quanyan12} that only gives a closed-form Nash solution for Homogeneous users, we obtain a closed-form Nash solution for a more general case.
\end{remark}

With the optimal response of the users, to characterize the Stackelberg equilibrium, one way is to view it as a finite horizon discrete-time optimal control problem. Another way is to form a nonlinear programming problem with objective function (\ref{SupplierCost2}) subject to the constraints provided by (\ref{NashSolState}-\ref{RiccatiLike2}). Here, we apply the second method and develop the following theorem to characterize the Stackelberg Equilibrium by utilizing the optimal response of users characterized in Theorem \ref{NashSol}.

\begin{theorem}
If $\{ \mathbf{s}^*\}$ denotes an open-loop Stackelberg equilibrium strategy for the supplier (leader) in the dynamic game in the form of (\ref{PriceDynamics2}) (\ref{UserCost2}) and (\ref{SupplierCost2}), and $\{\mathbf{d}^{i*},i\in\mathcal{N}\}$ is the corresponding best response of the users (followers), there exists vector sequences $\{\theta_1,...,\theta_T\}$, $\{\mu^i_1,...,\mu^i_T\}$, $\{\nu^i_1,...,\nu^i_T\}$ for $i\in\mathcal{N}$ such that the relations (\ref{NashSolCon}) (\ref{RiccatiLike2}) and the following relations are satisfied
\begin{equation}\label{StackState}
\begin{aligned}
p_{t+1}^*&=(1+\omega_t)p_t^*+\gamma(\sum\limits_{i\in\mathcal{N}} d_t^{i*} - s_t^*),\ \ \ p_1^*=p_1,
\end{aligned}
\end{equation}

\begin{equation}\label{StackLeaderCon}
\begin{aligned}
s_t^*&=\frac{p_t^* +\kappa_t d_t^*+\gamma \theta_t}{\bar{\delta}_t +\kappa_t }\\
&+\frac{[\gamma/(1+\gamma^2\sum_{j\in\mathcal{N}} q_{t+1}^j/\alpha_{t}^j)] \sum_{j\in\mathcal{N}}\mu_t^i }{\bar{\delta}_t +\kappa_t };\ \ \ \ \ \ \ \ \ 
\end{aligned}
\end{equation}

\begin{align}\label{StackLeaderCon2}
&\kappa_t(d_t^*-s_t^*)+\gamma \theta_t +\nu_t^i=0,\ \ \ \forall i\in\mathcal{N};\ \ \ \ \ \ \ \ \ \ \ \ \ \ \ 
\end{align}

\begin{equation}\label{StackTheta}
\begin{aligned}
    \theta_{t-1}&=(1+\omega_t)\theta_t - s_t^*-\sum\limits_{j\in\mathcal{N}} \frac{\nu_t^i (q_t^i -1)}{{\alpha_t^i}^2 (1 + \omega_t-\gamma/\alpha_t^i)},\\
    \theta_T&=0;
\end{aligned}
\end{equation}

\begin{equation}\label{StackMu}
\begin{aligned}
\mu_{t+1}^i&=\Big\{ 1+\omega_t -\frac{1}{\alpha_t^i}  -\frac{\gamma^2/\alpha_t^i}{1+\gamma^2\sum_{j\in\mathcal{N}} q_{t+1}^j/\alpha_t^j}    \Big\}\mu_t^i\\
&+\frac{\nu_{t+1}^i}{{\alpha_{t+1}^i}^2(1+\omega_{t+1}-\gamma/\alpha_{t+1}^i) },\\
\mu_{1}^i&=0,\ \ \ \forall i\in\mathcal{N}
\end{aligned}
\end{equation}
for all $t\in\mathbf{T}$, where $q_t^i,i\in\mathcal{N},t\in\mathbf{T}$ are given by (\ref{RiccatiLike}).
\end{theorem}
\begin{proof}
Since the problem is an optimal control problem in discrete-time, it equivalent to a finite dimensional nonlinear programming problem. To simplify the statement of the proof, we write equation (\ref{NashSolCon}) as $d_t^i=\Gamma^i_t(p_t,b_t^i)$, also we write equation (\ref{RiccatiLike2}) as $b_t^i=\Phi_t^i(b_{t+1}^i,\mathbf{b}_{t+1}^{-i})$. The Lagrangian of the nonlinear programming problem can be written as
\begin{equation}
\begin{aligned}
    L&=\Big\{ \sum\limits_{t=1}^T \frac{1}{2}\bar{\delta}_t s_t^2-s_t p_t + \frac{1}{2}\kappa_t (s_t-d_t)^2 \\
    &+\theta_t\{ (1+\omega_t)p_t +\gamma(d_t-s_t)\} \\
    &+\sum\limits_{i\in\mathcal{N}}\mu_t^i\{\Phi_t^i(b_{t+1}^i,\mathbf{b}_{t+1}^{-i})-b_{t}^i \}+\sum\limits_{i\in\mathcal{N}}\nu_t^i\{\Gamma_t^i(p_t,b_t^i)-d_t^i\} \Big\}
\end{aligned}
\end{equation}
where $\theta_t$, $\mu_t^i$ and $\nu_t^i$ for $i\in\mathcal{N},t\in\mathbf{T}$ denote appropriate Lagrangian multipliers. If $\{\mathbf{s}^*\}$ is a minimizing solution and $\{\mathbf{d}^{i*}_t, b_t^i,i\in\mathcal{N},p_t^*,t\in\mathbf{T}\}$ are the corresponding values of the other variables, so that the constraints in Theorem \ref{NashSol} are satisfied. Then it is necessary that 
\begin{equation*}
    \frac{\partial L}{\partial s_t}=0,\ \ \ \frac{\partial L}{\partial d_t}=0,\ \ \ \frac{\partial L}{\partial p_t}=0,\ \ \ \frac{\partial L}{\partial b_{t+1}}=0,\ \ \ t\in\mathbf{T},
\end{equation*}
wherefrom (\ref{NashSolCon}) (\ref{RiccatiLike2}) and (\ref{StackState})-(\ref{StackMu}).
\end{proof}

\begin{remark}
Note that $\gamma$, $\omega_t$, $\kappa_t$, $\bar{\delta}_t$, $\alpha_t^i$, $\bar{\psi_t^i}$ for $i\in\mathcal{N},t\in\mathbf{T}$ are system parameters that are known. And $q_t^i$ for $i\in\mathcal{N},t\in\mathbf{T}$ can be computed independently based on the rule (\ref{RiccatiLike}). Then the relations (\ref{NashSolCon}) (\ref{RiccatiLike2}) together with the relations (\ref{StackState}-\ref{StackMu}) form a set of linear constraints that characterizes the Stackelberg equilibrium of the linear-quadratic game. There are $(4N+3)T$ constraints and $(4N+3)T$ unknowns where $N$ is the number of users. From (\ref{StackLeaderCon}), we know that as the price goes up, the supplier tends to generate more energy to gain more profit. Also, the supply increases as the demand increases.
\end{remark}

\section{Simulation}\label{Sec:Sim}

In this section, we present simulation results to show the demand response and the supply response under market-driven dynamic price. We assume that there are $8$ users in the market, i.e., $N=8$. Let $\gamma=0.05$, $\omega_t=-0.05$. Suppose $\alpha_t^1=\alpha_t^2=\alpha=0.5,\alpha_t^3=\alpha_t^4=\alpha_t^5=0.6,\alpha_t^6=\alpha_t^7=\alpha_t^8=0.7$ for all $t\in\mathbf{T}$ and $\kappa_t=\kappa=0.3$ for all $t\in\mathbf{T}$. We consider the horizon $T=20$ days. Suppose each user have $3-5$ appliances including refrigerator-freezer (daily usage: 1.32kWh), electric stove (daily use: 7.89kWh), heating (daily usage: 7.1KWh), lighting (daily usage: 1KWh), Home electronics (daily usage: 0.5 KWh). We assign weights $\alpha^{i,a}$ proportional to the level of daily usage and normalize them between $0$ to $q$. And we set $q=5$. \begin{figure}
\centering
  \includegraphics[width=8cm]{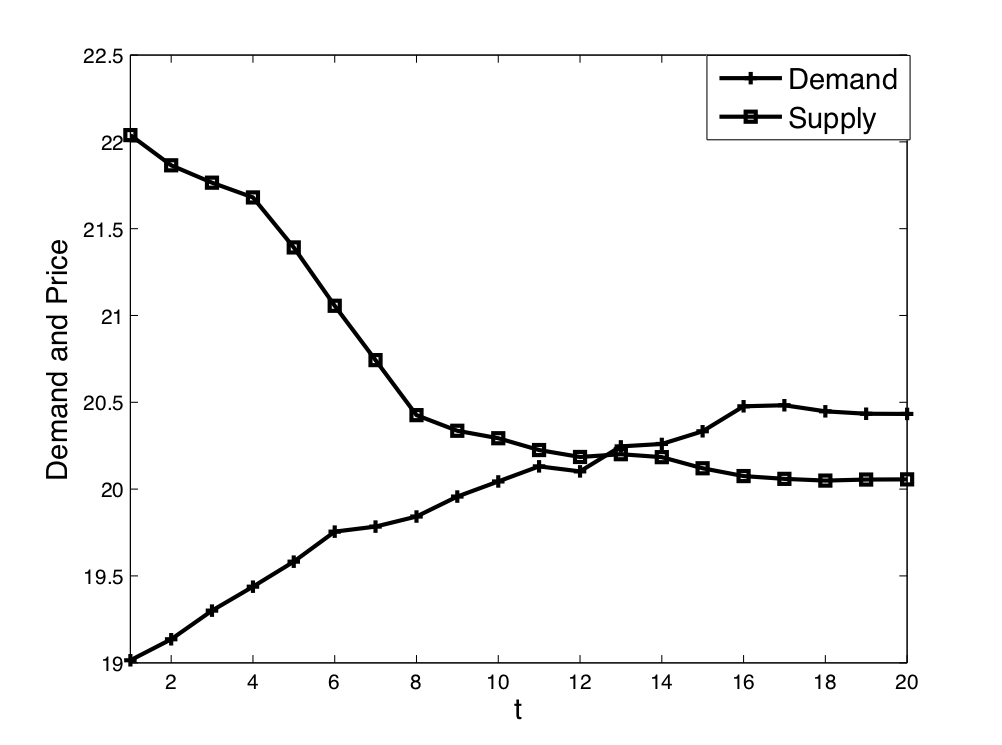}
  \caption{The total demand of all users $d_t$ and the supply $s_t$ for $t=1,...,20$}
  \label{fig:DemandSupply}
\end{figure}

Suppose the supplier are equipped with $3$ generators including fixed gas thermal generator, solar generator and hydroelectric generator with cost parameters $\delta_t^1=1.5$, $\delta_t^2=0.8$, $\delta_t^3=1$ dollar/KWh$^2$. Choose initial price $p_1=1.5$ dollar/KWh. In Fig. \ref{fig:DemandSupply}, we present the optimal supply and the optimal demand response from time $1$ to time $20$ based on the linear relations that characterize the Stackelberg equilibrium. We can see that the supply and the demand is not balanced at the beginning days. The demand and the supply are tuned by the dynamic price and get closer. The supply does not exactly match with the demand but they are almost balanced in the later days. The evolution of the price from time $1$ to time $21$ is presented in Fig. \ref{fig:Price}. We set the initial price as $1.5$. The price goes all way down to $0.26$ which indicates the market drives the price to a reasonable value.

\begin{figure}[H]
\centering
  \includegraphics[width=8cm]{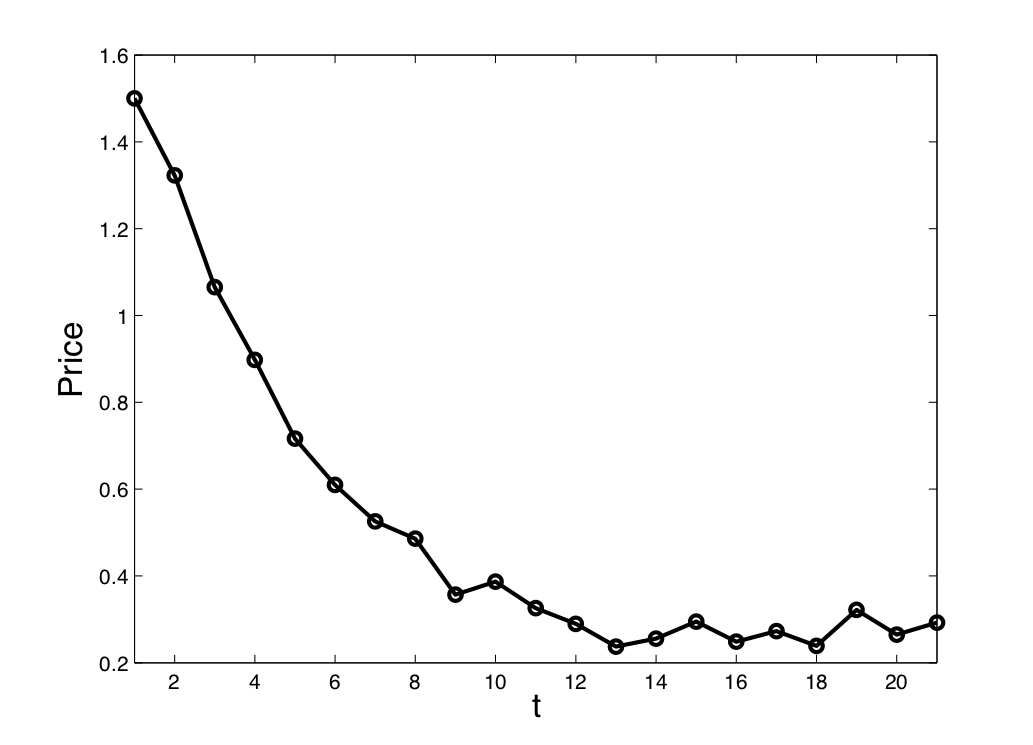}
  \caption{The evolution of the price $p_t$ over time $t=1,...,21$.}
  \label{fig:Price}
\end{figure}

\section{Conclusion}\label{Sec:Con}
In this paper, we have jointly studied the optimal supply and the distributed demand response under market-driven dynamic price. A Stackelberg dynamic game has been utilized to model the intertwined interactions among price, demand and supply. At the lower level, a set of static optimization problems has been formulated to solve the optimal power allocation problems for users and the optimal generation allocation problems for the supplier. Given the supply schedule, we characterize the optimal response from the uses by characterizing the Nash equilibrium solution using a set of necessary conditions. Further, the Stackelberg equilibrium has been characterized by a set of linear constraints. Simulation results have shown that the supply and the demand are tuned by the dynamic price to be balanced.

Future work can focus on studying the case with multiple energy providers and multiple users. Also, a Stackelberg dynamic game with a large population of followers and how the response would be when the number of followers $N$ approaches infinity are also of interest. For the lower level,  we can include storage devices and study their impact on the generation cost, users' utility and the entire power grid.


%

\ifCLASSOPTIONcaptionsoff
  \newpage
\fi

\bibliography{references}

\begin{thebibliography}{10}
\providecommand{\url}[1]{#1}
\csname url@samestyle\endcsname
\providecommand{\newblock}{\relax}
\providecommand{\bibinfo}[2]{#2}
\providecommand{\BIBentrySTDinterwordspacing}{\spaceskip=0pt\relax}
\providecommand{\BIBentryALTinterwordstretchfactor}{4}
\providecommand{\BIBentryALTinterwordspacing}{\spaceskip=\fontdimen2\font plus
\BIBentryALTinterwordstretchfactor\fontdimen3\font minus
  \fontdimen4\font\relax}
\providecommand{\BIBforeignlanguage}[2]{{%
\expandafter\ifx\csname l@#1\endcsname\relax
\typeout{** WARNING: IEEEtran.bst: No hyphenation pattern has been}%
\typeout{** loaded for the language `#1'. Using the pattern for}%
\typeout{** the default language instead.}%
\else
\language=\csname l@#1\endcsname
\fi
#2}}
\providecommand{\BIBdecl}{\relax}
\BIBdecl

\bibitem{Mohsenian10}
A.-H. Mohsenian-Rad, V.~W. Wong, J.~Jatskevich, R.~Schober, and A.~Leon-Garcia,
  ``Autonomous demand-side management based on game-theoretic energy
  consumption scheduling for the future smart grid,'' \emph{IEEE transactions
  on Smart Grid}, vol.~1, no.~3, pp. 320--331, 2010.

\bibitem{Quanyan12}
Q.~Zhu, Z.~Han, and T.~Ba{\c{s}}ar, ``A differential game approach to
  distributed demand side management in smart grid,'' in \emph{2012 IEEE
  International Conference on Communications (ICC)}.\hskip 1em plus 0.5em minus
  0.4em\relax IEEE, 2012, pp. 3345--3350.

\bibitem{Kamyab16}
F.~Kamyab, M.~Amini, S.~Sheykhha, M.~Hasanpour, and M.~M. Jalali, ``Demand
  response program in smart grid using supply function bidding mechanism,''
  \emph{IEEE Transactions on Smart Grid}, vol.~7, no.~3, pp. 1277--1284, 2015.

\bibitem{Paola17}
A.~De~Paola, D.~Angeli, and G.~Strbac, ``Price-based schemes for distributed
  coordination of flexible demand in the electricity market,'' \emph{IEEE
  Transactions on Smart Grid}, vol.~8, no.~6, pp. 3104--3116, 2017.

\bibitem{Ghasemkhani18}
A.~Ghasemkhani and L.~Yang, ``Reinforcement learning based pricing for demand
  response,'' in \emph{2018 IEEE International Conference on Communications
  Workshops (ICC Workshops)}.\hskip 1em plus 0.5em minus 0.4em\relax IEEE,
  2018, pp. 1--6.

\bibitem{Li11}
N.~Li, L.~Chen, and S.~H. Low, ``Optimal demand response based on utility
  maximization in power networks,'' in \emph{2011 IEEE power and energy society
  general meeting}.\hskip 1em plus 0.5em minus 0.4em\relax IEEE, 2011, pp.
  1--8.

\bibitem{Juntao17}
J.~Chen and Q.~Zhu, ``A game-theoretic framework for resilient and distributed
  generation control of renewable energies in microgrids,'' \emph{IEEE
  Transactions on Smart Grid}, vol.~8, no.~1, pp. 285--295, 2016.

\bibitem{Maharjan13}
S.~Maharjan, Q.~Zhu, Y.~Zhang, S.~Gjessing, and T.~Basar, ``Dependable demand
  response management in the smart grid: A stackelberg game approach,''
  \emph{IEEE Transactions on Smart Grid}, vol.~4, no.~1, pp. 120--132, 2013.

\bibitem{JuntaoSt17}
J.~Chen and Q.~Zhu, ``A stackelberg game approach for two-level distributed
  energy management in smart grids,'' \emph{IEEE Transactions on Smart Grid},
  vol.~9, no.~6, pp. 6554--6565, 2017.

\bibitem{Brigo03}
D.~Brigo, F.~Mercurio, G.~Sartorelli \emph{et~al.}, ``Alternative asset-price
  dynamics and volatility smile,'' \emph{Quantitative Finance}, vol.~3, no.~3,
  pp. 173--183, 2003.

\bibitem{DeJong06}
C.~De~Jong, ``The nature of power spikes: A regime-switch approach,''
  \emph{Studies in Nonlinear Dynamics \& Econometrics}, vol.~10, no.~3, 2006.

\bibitem{Hou05}
K.~Hou and T.~J. Moskowitz, ``Market frictions, price delay, and the
  cross-section of expected returns,'' \emph{The Review of Financial Studies},
  vol.~18, no.~3, pp. 981--1020, 2005.

\bibitem{Luenberger97}
D.~G. Luenberger, \emph{Optimization by vector space methods}.\hskip 1em plus
  0.5em minus 0.4em\relax John Wiley \& Sons, 1997.

\bibitem{Yaiche00}
H.~Ya{\"\i}che, R.~R. Mazumdar, and C.~Rosenberg, ``A game theoretic framework
  for bandwidth allocation and pricing in broadband networks,'' \emph{IEEE/ACM
  transactions on networking}, vol.~8, no.~5, pp. 667--678, 2000.

\bibitem{Kelly98}
F.~P. Kelly, A.~K. Maulloo, and D.~K. Tan, ``Rate control for communication
  networks: shadow prices, proportional fairness and stability,'' \emph{Journal
  of the Operational Research society}, vol.~49, no.~3, pp. 237--252, 1998.

\bibitem{TBBook}
T.~Basar and G.~J. Olsder, \emph{Dynamic noncooperative game theory}.\hskip 1em
  plus 0.5em minus 0.4em\relax Siam, 1999, vol.~23.

\bibitem{Bertsekas05}
D.~P. Bertsekas, D.~P. Bertsekas, D.~P. Bertsekas, and D.~P. Bertsekas,
  \emph{Dynamic programming and optimal control}.\hskip 1em plus 0.5em minus
  0.4em\relax Athena scientific Belmont, MA, 1995, vol.~1, no.~2.

\end{thebibliography}
\bibliographystyle{IEEEtran}
%




\end{document}